\documentclass[10pt,a4paper]{article}

\usepackage{fullpage}

\usepackage{ifthen}
\usepackage{calc}

\usepackage[english]{babel}

\usepackage[centertags,reqno]{amsmath}
\usepackage{amsthm}
\usepackage{amsfonts}
\usepackage{amssymb}
\usepackage{amsmath}
\usepackage{eepic}
\usepackage{eepicemu}
\usepackage{eso-pic}
\usepackage{epsfig}
\usepackage{graphicx}
\usepackage{pstricks}
\usepackage{caption}
\usepackage{subcaption}

\usepackage{float}

\usepackage{tikz}
\usetikzlibrary{matrix,fit,decorations.markings,arrows,shapes,positioning,decorations.pathmorphing,fit}

\usepackage{chngpage}

\usepackage[colorlinks=false, linktocpage=true, linkbordercolor={white}]{hyperref} 

\usepackage[colorinlistoftodos]{todonotes}

\usepackage{url}

\theoremstyle{plain}
\newtheorem{theorem}{Theorem}[section]
\newtheorem{lemma}[theorem]{Lemma}

\newtheorem{corollary}[theorem]{Corollary}

\newtheorem{problem}[theorem]{Problem}
\theoremstyle{definition}
\newtheorem{definition}[theorem]{Definition}

\usepackage{lineno} 
 \newcommand{\induce}[2]{\mbox{$ #1 \langle #2 \rangle$}}

\begin{document}
\bibliographystyle{plain}
\title{DAG-width and circumference of digraphs}
  \author{J\o{}rgen Bang-Jensen and Tilde My Larsen\thanks{Department of Mathematics and Computer Science, University of Southern Denmark, Denmark}}
  \date{\today}
\maketitle 

\begin{abstract}
We prove that every digraph of circumference $l$ has DAG-width at most $l$ and this is best possible. As a consequence of our result we deduce that the 
$k$-linkage problem is polynomially solvable for every fixed $k$ in the class
of digraphs with bounded circumference. This answers a question posed in \cite{bangTCS562}. We also prove that the weak $k$-linkage problem (where we ask for arc-disjoint paths) is polynomially solvable for every fixed $k$
in the class of digraphs with circumference 2 as well as for digraphs with a bounded number of disjoint cycles each of length at least 3. The case of bounded circumference digraphs is open. Finally we prove that the minimum spanning strong subdigraph problem is NP-hard on digraphs of DAG-width at most 5.\\

\noindent{}{\bf Keywords:} DAG-width, k-linkage problem, bounded cycle length, polynomial algorithm, cops and robber game
\end{abstract}

\section{Introduction}

Terminology and  notation not described below follows \cite{bang2009}.
A digraph $D=(V,A)$ has vertex set $V$ and arc set $A$. The {\bf out-degree} ({\bf in-degree}), denoted $d^+(v)$ ($d^-(v)$), of a vertex $v$ is the number of arcs from $v$ to $V-v$ (from $V-v$ to $v$). For $X,Y\subset V$ with $X\cap Y=\emptyset$ an $(X,Y)$-path is a directed path starting in $X$ and ending in $Y$ and with all other vertices in 
$V-(X\cup Y)$. For a directed path or cycle  $P$ containing vertices $u$, $v$ with $u$ before $v$ on $P$ we denote by $P[u,v]$ the subpath of $P$ from $u$ to $v$. We use the notation $[k]$ for the set $\{1,2,\ldots{},k\}$. A $k$-cycle is a directed cycle with $k$ vertices. The {\bf circumference} of a digraph $D$ is the length of a longest directed cycle in $D$. A directed acyclic graph (DAG) is a digraph without directed cycles.

For a given natural number $k$ the $k$-linkage problem is as follows: Given a digraph $D$ and $2k$ distinct vertices $s_1,\ldots{},s_k,t_1,\ldots{}t_k$ (called {\bf terminals}); determine whether $D$ has $k$ disjoint paths $P_1,\ldots{},P_k$ such that $P_i$ is an $(s_i,t_i)$-path for $i\in [k]$.
While the undirected analogue of the $k$-linkage problem is polynomial for every fixed $k$ by the Robertson-Seymour theorem \cite{robertsonJCT63}, the directed version is NP-complete already for $k=2$ \cite{fortuneTCS10}. The problem is known to be polynomially solvable for fixed $k$ when $D$ belongs to one of the following classes of digraphs: acyclic digraphs \cite{fortuneTCS10}, semicomplete digraphs \cite{bangSJDM5,chudnovskysub}, digraphs of bounded directed-tree-width \cite{johnsonJCT82} (this includes digraphs of bounded DAG-width which will be defined later), digraphs of bounded Kelly-width \cite{hunterTCS399}  and $d$-path-dominant digraphs \cite{chudnovskysub}. A digraph $D$ is $d$-path-dominant for some $d\geq 1$ if every minimal path\footnote{A path $P$ from a vertex $x$ to a vertex $y$ is {\bf minimal} if there is no $(x,y)$-path $P'$ such that $V(P')$ is a proper subset of $V(P)$.} $P$ of $D$ with $d$ vertices has the property that there is at least one arc between every vertex of $D-V(P)$ and $V(P)$. Thus the 1-path-dominant digraphs are the semicomplete digraphs (every pair of distinct vertices have at least one arc between them).

In \cite{bangTCS562} it was asked whether the $k$-linkage problem would be polynomially solvable for digraphs of bounded circumference (see also \cite{havetsub}). This can be seen as a generalization of the result by Fortune et al for acyclic digraphs. In this paper we answer the question in the affirmative by showing that digraphs of bounded circumference have bounded DAG-width (defined below) and hence also bounded directed tree-width. From this the result follows since the $k$-linkage problem is polynomial for digraphs of bounded directed tree-width \cite{johnsonJCT82}. 

We also consider the weak $k$-linkage problem where we ask for arc-disjoint paths $P_1,\ldots{},P_k$ such that $P_i$ is an $(s_i,t_i)$-path for $i\in [k]$. Now we may have $|\{s_1,\ldots{},s_k,t_1,\ldots{},t_k\}|<2k$ (e.g. $s_2=s_6=t_1$).
We prove that for every fixed natural number $k$ the weak $k$-linkage problem is polynomially solvable for digraphs of circumference 2, digraphs with no closed trail longer than some constant and for digraphs having no set of $l$ disjoint cycles each of length at least 3 (the proof of the latter uses a result from \cite{havetsub}, see Theorem \ref{EPdigraphs}).

\section{DAG-width: definitions and some results}
Robertson and Seymour introduced the concept  tree-width of undirected graphs. The tree-width  measure has many nice properties, including polynomial solutions for many NP-complete problems on graphs of bounded tree-width. Several attempts have been made to find a measure for directed graphs with similar properties as tree-width for undirected graphs. Unfortunately there is evidence \cite{ganianLNCS6478} that in some sense none can exist,  but several measures that work nicely on certain problems have been made. In the following we will consider one of these, the DAG-width of a directed graph. The theory of this section is based on 
\cite{obdrzalekSODA2006,berwangerJCT102}. We will start by defining the concept of DAG-width and then relate this to a directed version of the cops and robbers game. As the name suggests, the DAG-width of a digraph is a measure of how close it is to being acyclic. 

\begin{definition}
For an acyclic digraph $D=(V,A)$ we define the partial ordering $\leq_D$ to be the reflexive, transitive closure of the arcs of $D$. A {\bf root} of a set $X \subseteq V$ is then a $\leq_D$-minimal element of $X$ and analogously a {\bf leaf} of a set $X \subseteq V$ is a $\leq_D$-maximal element.
\end{definition}

\begin{definition}
 Let $D=(V,A)$ be a directed graph with vertex set $V$ and arc set $A$. We say a set $W \subseteq V$ {\bf guards} a set $V' \subseteq V$ if every arc leaving $V'$ is incident with $W$, i.e. for all $uv\in A$ with $u \in V'$ and $v \notin V'$ we will have $v \in W$. 
\end{definition}

We can now define the DAG-width. 
\begin{definition}\cite{obdrzalekSODA2006}
Let $D=(V,A)$ be a directed graph. Then a {\bf DAG-decomposition} of $D$ is a pair $(H, \mathcal{X})$, where $H$ is a DAG and $\mathcal{X}=(X_h)_{h \in V(H)}$ such that
\begin{description}
  \item[\textnormal{D1})] $\bigcup_{h \in V(H)} X_h = V(D)$
  \item[\textnormal{D2)}] For all vertices $h,h',h''\in V(H)$ such that 
 $h \leq_H h' \leq_H h''$ we have $X_h \cap X_{h''}\subseteq X_{h'}$
  \item[\textnormal{D3)}] For all arcs $(h,h')\in A(H)$, $X_h \cap X_{h'}$ guards $X_{\geq {h'}} \backslash X_h$ where $X_{\geq {h'}} = \cup_{h'\leq_H h''} X_{h''}$. For any root $h$, the set $X_{\geq h}$ is guarded by $\emptyset$.
\end{description}

The {\bf width} of the DAG-decomposition $(H,\mathcal{X})$ is defined as $\max \{|X_h| : h \in V(H)\}$ and the {\bf DAG-width} of $D$ is the minimum width over all DAG-decompositions of $D$. The DAG-width of an acyclic graph is 1 (just let $V(H)=V(D)$ and $X_h=\{h\}$ for each $h\in V(H)$). 
 \end{definition}

Notice that given a DAG-decomposition $(H,\mathcal{X})$ of $D$ which does not have a unique root, we can always add the empty set to $\mathcal{X}$ and a corresponding new vertex to $H$, such that the new vertex has an arc to every root of $H$. This will give a new DAG-decomposition of $D$, with the same width as $(H,\mathcal{X})$, but with a unique root. Hence we may always assume that given a DAG-decomposition of $D$, it has  a unique root. It can also be shown that the DAG-width of a digraph is equal to the maximum of the DAG-widths of its strongly connected components \cite{obdrzalekSODA2006}.

\begin{theorem}[\cite{obdrzalekSODA2006}] It is NP-hard to decide
 for inputs $D$ and $k$ whether the DAG-width of  $D$ is at most $k$. 
\end{theorem}

We will now give two results that are related to another measure, called the directed tree-width \cite{johnsonJCT82}. As the definition  is quite technical and not needed for our purposes, we shall not give it here.

\begin{theorem}[\cite{berwangerJCT102}]\label{ThmDirTree}
If a graph has DAG-width $k$ then it has directed tree-width at most $3k+1$.
\end{theorem}

\begin{theorem}[\cite{johnsonJCT82}]
For every fixed natural number $k$ the $k$-linkage problem is polynomially solvable for directed graphs of bounded directed tree width.
\end{theorem}
 
\begin{corollary}[\cite{obdrzalekSODA2006}]\label{CorBoundDag}
 For every fixed natural number $k$ the $k$-linkage problem is polynomially solvable for directed graphs of bounded DAG-width.\qed
\end{corollary}

A useful way to obtain a  DAG-decomposition of a given digraph is via
the game of cops and robber. This game was first introduced by Seymour and Thomas and used in the study of the tree-width in undirected graphs. For the directed case, different variations of the game have been studied. We will pose the game in the context of the DAG-width \cite{obdrzalekSODA2006}. The principle of the game is that a robber is moving around  the digraph from vertex to vertex respecting the orientation of the arcs. The robber can run infinitely fast and wants to avoid getting caught by the cops. He can take any path from his current vertex to another, provided no intermediate vertex on that path is currently occupied by a cop.  The cops can either stand on  vertices of the digraph or  be in a helicopters above the graph (the point of the helicopters is that cops are not constrained to move along paths in the digraph). The cops win if they can land on the vertex occupied by the robber. We say that the cops have a {\bf cop-monotone} strategy if the cops never visit a vertex in the graph more than once. Similarly, the cops have a {\bf robber-monotone} strategy if the the set of vertices the robber can reach (without running through a cop hosting vertex) is non increasing. 


\begin{lemma}\cite{berwangerJCT102})\label{lemmaCopRobber}
If the cop player has a cop-monotone or robber-monotone winning
strategy then he also has a winning strategy that is both cop- and robber-monotone.
\end{lemma}

The strategy for the cops  will be to split the graph into strongly connected components in such a way that the set of vertices reachable by the robber (without running into a cop) becomes smaller and smaller until the cops can finally land on the vertex occupied by the robber. Before stating the essential connection between DAG-width and the game of cops and robber, it is helpful to
consider the connection between guarding sets and the game of cops and robber. If a robber is in a set $V'$ and there is a cop on every vertex in $W$ for some guarding set $W$ of $V'$, then the robber cannot leave $V'$ without running to a vertex with a cop and hence getting caught. Now if there are  cops not placed on any vertex, these can land on vertices in $V'$ and hence forcing the robber into a smaller piece of the digraph. 

\begin{theorem}\cite{berwangerJCT102,obdrzalekSODA2006}
\label{ThmCopsVSDag}
A directed graph $D$ has DAG-width $k$ if and only if it takes $k$ cops to catch the robber using a cop-monotone strategy.
\end{theorem}

\section{Digraphs with bounded circumference}
Using the game of cops and robber we are now ready to prove our main result. A cycle $C$ is {\bf maximal} in $D$ if $D$ has no cycle properly containing the vertices of $C$. Note that in a digraph with bounded circumference one can check maximality of any cycle $C$ and find a larger cycle containing $V(C)$ if one exists in polynomial time.
 
\begin{theorem}
\label{main}
 Let $D$ be a directed graph with circumference at most $p\in \mathbb{N}$. Then the DAG-width of $D$ is at most $p$ and this is best possible.
\end{theorem}

\begin{proof}
Since the DAG-width of a digraph is the maximum of the DAG-widths of its strong components, it suffices to consider the case when $D$ is strong. 

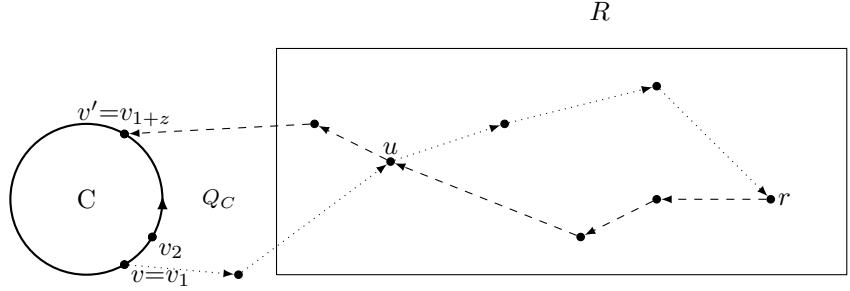
\begin{figure}
 \centering
\begin{tikzpicture}[scale=0.5]
  \node[circle] () at (0,2){C};
  \draw[font=\footnotesize](3.5,2) node {$Q_C$};	
  \draw[](13.5,7) node {$R$};	
  \tikzstyle{every node}=[fill=black, inner sep=1pt]
  
  \draw[thick] (0,2cm) circle(2);
  \foreach \x in {-60,-30,60} {
      \node[circle, draw,yshift=1cm] (v\x) at (\x:2) {};
  } 

  \node[draw,circle,label=below right:$v{=}v_1$] (v1) at (v-60){};
  \node[draw,circle,label=below right:$v_2$] (v2new) at (v-30){};
  \node[draw,circle,label=above:$v'{=}v_{1+z}$] (v2) at (v60){};
  
  \node[circle, draw, label=right:$r$] (r) at (18,2){};
  \node[circle, draw] (v11) at (4,0){};
  \node[circle, draw, label=above:$u$] (v12) at (8,3){};
  \node[circle, draw] (v111) at (11,4){};
  \node[circle, draw] (v1111) at (15,5){};
  \node[circle, draw] (v22) at (6,4){};
  \node[circle, draw] (v222) at (13,1){};
  \node[circle, draw] (v2222) at (15,2){};
  \draw[-latex,dotted] (v1) to (v11);
  \draw[-latex,dotted] (v11) to (v12);
  \draw[-latex,dotted] (v12) to (v111);
  \draw[-latex,dotted] (v111) to (v1111);
  \draw[-latex,dotted] (v1111) to (r);
  \draw[-latex,dashed] (r) to (v2222);
  \draw[-latex,dashed] (v2222) to (v222);
  \draw[-latex,dashed] (v222) to (v12);
  \draw[-latex,dashed] (v12) to (v22);
  \draw[-latex,dashed] (v22) to (v2);
  \draw[-latex,thick] (2,2) to (2,2.1);
  \draw[] (5,0) -- (5,6)--(20,6)--(20,0)--(5,0); 
   
\end{tikzpicture}
\caption{dotted arcs form  the $P_{C,r}$ path and dashed the $P_{r,C}$ path. The cycle $Q_C$ is formed by the three subpaths $P_{r,C}[u,v_{z+1}],C[V_{z+1},v_1],P_{C,r}[v_1,u]$.}\label{Qc}
\end{figure}

Below we describe a robber-monotone winning strategy for  $p$ cops. We start by having no cops in the digraph. Now pick an arbitrary maximal cycle $C$ of $D$ and put $|C|$ cops on the vertices of $C$. If the robber is not caught already, he will
have moved to a vertex $r$ of $V-V(C)$. Let us denote by $R\subseteq V-V(C)$ the set of vertices the robber can reach without running into a cop, that is the set of vertices reachable from $r$ in $D-V(C)$.
 As $D$ is strong,  there must be a $(V(C),r)$-path $P_{C,r}$ starting in some 
 vertex $v$ of  $C$ and ending in  $r$. Similarly, there must be an $(r,V(C))$-path $P_{r,C}$ from $r$ to some vertex $v'$ of $C$. 
Their concatenation
$P_{C,r}P_{r,C}$ will be a (possibly closed) $(v,v')$-trail  in $D$.
 Now let $u$ be the first vertex on the trail with $u\in V(P_{C,r})\cap V(P_{r,C})-V(C)$ (such a vertex exists as $r$ is on both paths). Then we let $Q_C=P_{C,r}[v,u]P_{r,C}[u,v']$ be the path (a cycle when $v=v'$) in $D$ formed by following $P_{C,r}$ to $u$ and then proceeding from $u$ to $v'$ on $P_{r,C}$. Note that $u\in R$ as $u$ can be reached by $r$ in $V-V(C)$. See Figure \ref{Qc}.

Let us first consider the case where there exists a pair of  paths $P_{C,r}$, $P_{r,C}$ such that $Q_C$ is a path (that is, $v\neq v'$). Pick the paths $P_{C,r},P_{r,C}$ such that $|V(C[v,v'])|>0$ is as small as possible and let $v_1,v_2,\ldots{},v_{|C|}$ be an enumeration of the  vertices in $C$ 
such that and $v_1=v$ and $v_iv_{(i+1)\mod|C|}$ is an arc in $C$ for $i\in [|V(C)|]$. Let $z$ be the integer such that $v'=v_{1+z}$ and note that, by our choice of paths, $v_{1+z}$ is the first vertex on the cycle $C$ after $v_1$ that is the end vertex of some  $P_{r,C}$ path. 
Now $C'=C[v_{1+z},v_1]Q_C$ is a cycle, containing at least one vertex from $R$ (namely $u$).
Observe that $z\neq 1$, since otherwise the cycle $C'$ would be longer than $C$, contradicting the maximality of $C$.   Move 
(some of) the cops currently on the vertices of $C[v_{2},v_{z}]$ to the currently unoccupied vertices of $C'$. Further if there are still unoccupied vertices on $C'$ (when $|V(C[v_{2},v_{z}])| < |V(Q_C)|$), then place cops currently in helicopters on these. Conversely, if there are more available cops from 
$C[v_{2},v_{z}]$  than needed to cover the $C'$ vertices 
(when $|V(C[v_{2},v_{z}])| > |V(Q_C)|$), then place these in helicopters for later use. Notice that we will never run out of cops, since all our placed cops are on the cycle $C'$, and $|C'|\leq p$. 
Now extend $C'$ to a maximal cycle $C^*$ containing all the vertices of $C'$. After covering possible new vertices in $C^*$, by cops from helicopters, the strategy can be repeated starting with cops on $V(C^*)$ and we will have $R^*\subsetneq R$, where $R^*$ is the set of vertices reachable by the robber after removing $V(C^*)$.


It remains to consider the case where the paths $P_{C,r}$ and $P_{r,C}$  are incident to the same vertex $v$ in $C$ for all choices of such  paths, and hence every $Q_C$  a cycle. In this case, every cop except the one in $v$ is free to be lifted as every path from $R$ to $C$ enters $C$ in $v$. Take an arbitrary $Q_C$ cycle, and as in the case above, find a maximal  cycle $C^{**}$ containing the vertices of $Q_C$ and occupy each of these vertices with a cop. Again since the length $C^{**}$ is at most $p$, we have enough cops. Now we can repeat the strategy above starting from the new maximal occupied cycle.

The  strategy described above is indeed a robber-monotone strategy: When moving the cops from $C$ to $C'$, we only move cops that cannot be reached by the robber without the robber running through another cop occupied vertex, hence we do not open up a new part of the graph for the robber, and hence $R'\subseteq R$ will always hold. Furthermore, for each new maximal cycle we occupy in the strategy, at least one of these vertices belongs to the current set $R$. Hence in at most $|V|$ steps $R=\emptyset$ and the robber is caught so the strategy above is a winning one.

To see that $p$ cops may be necessary to catch a robber in a digraph with circumference $p$ consider the complete digraph on $p$ vertices, that is, there is an arc on both directions between any pair of distinct vertices. If we only have $p-1$ cops here there will always be a free vertex and the robber (who moves infinitely fast) can move directly to that vertex from his current position as soon as the cops have announced their new positions. 

\end{proof}

\begin{corollary}
For every natural number $p$ there exists an algorithm ${\cal A}_p$ which given a digraph of circumference $p$ outputs a DAG-decomposition of width $p$ in polynomial time.
\end{corollary}

\begin{proof}
This follows from the fact that we can translate a robber monotone strategy for $p$ cops into a DAG-decomposition of width at most $k$ \cite{obdrzalekSODA2006}. See an example in the appendix at the end of the paper.
Note that the steps in the proofs above, such as finding a maximal  cycle containing a given set $X$ of vertices, are all polynomial since the length of the cycle sought is at most $p$ and hence we may check all possible cycles covering $X$ in polynomial time.\\
\end{proof}

Combining Theorem \ref{ThmDirTree} and Theorem  \ref{main} we obtain the following.

\begin{corollary}
\label{shortcyclesmalldtw}
Every digraph of circumference  most $l$ has directed tree-width at most $3l+1$. 
\end{corollary}

\begin{corollary}\cite{birmeleJGT43}.
If an undirected graph $G$ has circumference at most $k$, then its tree-width is at most $k-1$.
\end{corollary}

\begin{proof}
Given $G$ we form the digraph $\stackrel{\leftrightarrow}{G}$ by replacing each edge of $G$ by a 2-cycle. Then the tree-width of $G$ equals the DAG-width of
 $\stackrel{\leftrightarrow}{G}$ minus one \cite{berwangerJCT102} and now the claim follows from Theorem \ref{main}.
\end{proof}

\section{Linkings in digraphs with bounded circumference}

The following direct consequence of Corollary \ref{CorBoundDag} and Theorem \ref{main} answers a question in \cite{bangTCS562} in the affirmative.

\begin{theorem}
\label{klinkboundedcirc}
For every choice of natural numbers $k,l$ there exists a polynomial algorithm for the $k$-linkage problem on digraphs with circumference at most $l$.\qed
\end{theorem}

The case $l=2$ was proved previously in 
\cite{havetsub}. That  proof  uses both the polynomial algorithm for $k$-linkage in undirected graphs from 
\cite{robertsonJCT63} and an algorithm similar to that used in the algorithm for acyclic digraphs \cite{fortuneTCS10} to obtain an algorithm of running time roughly $O(n^{2k})$ where $n$ is the number of vertices in the digraph.

The reduction below to the $k$-linkage problem for acyclic digraphs leads to a faster algorithm
 since the complexity of the overall algorithm will be the same as the complexity of solving $k$-linking in DAGs which is $O(k!n^{k+2})$ (see e.g. \cite[Section 10.4]{bang2009}).

\begin{theorem} \label{Theorem2CycleVertex}
 Let $D$ be a directed graph of circumference 2 and let $k$ be a fixed integer. Then the $k$-linkage problem in $D$ can be reduced to a $k$-linkage problem in an acyclic graph in linear time.
\end{theorem}

\begin{proof}

First observe that since we are looking for disjoint paths, we may assume that $d^-(s_i)=d^+(t_j)=0$ for $i,j\in [k]$. Also note that each 
strongly connected
component of $D$ will correspond to a tree $T$ in the underlying undirected graph where each edge $uv$ of $T$ is replaced by a 2-cycle on $u,v$.

We will transform an instance $[D,s_1,\ldots{},s_k,t_1,\ldots{},t_k]$ to an equivalent instance 
$[D',s_1,\ldots{},s_k,t_1,\ldots{},t_k]$ where $D'$ is an acyclic digraph by removing one 2-cycle at a time, until there are no cycles left in the graph. The terminals stay the same throughout the reduction. Let $S$ be an arbitrary non-trivial strong component of $D$ (if none exists, $D$ is acyclic and we are done). By the assumption 
above, $S$ does not contain any of the vertices $s_1,\ldots{},s_k,t_1,\ldots{},t_k$ and hence none of the paths in a $k$-linkage can start or end in $S$.

\begin{figure}[H]
\centering
 \begin{minipage}{6 cm}
    \begin{tikzpicture}[scale=0.5]
    \tikzstyle{every node}=[font=\footnotesize]
    \node[circle, draw] (v) at (5,5){$v$};
    \draw (5,16) -- (0,4) --(10,4)--(5,16);
    \node[font=\large] at (5,12){$S$};
    \node[circle,draw] (u) at (5,10){$u$};
    \node[circle] (a1) at (0,2){\phantom{u}};
    \node[circle] (a2) at (0,15){\phantom{u}};
    \node[circle] (a3) at (10,2){\phantom{u}};
    \node[circle] (a4) at (10,15){\phantom{u}};
    \node[circle] (b1) at (4.5,12){};
    \node[circle] (b2) at (5.5,12){};
    \draw[-latex] (a1)--(v);
    \draw[-latex] (a2) --(b1);
    \draw[-latex] (v) --(a3);
    \draw[-latex] (b2) --(a4);
    \draw[-latex] (v) to [in=240, out=120] (u);
    \draw[-latex] (u) to [in=50, out=310] (v);
    \end{tikzpicture}
    \caption*{}
 \end{minipage}
  \begin{minipage}{6 cm}
  \begin{tikzpicture}[scale=0.5]
    \tikzstyle{every node}=[font=\footnotesize]
    \node[circle, draw] (v1) at (18,5){$v_1$};
    \node[circle, draw] (v2) at (22,5){$v_2$};
    \draw (20,16)--(17,8)--(23,8)--(20,16);
    \node[font=\large] at (20,11) {$S-v$};
    \node[circle,draw] (u) at (20,10){$u$};
    \node[circle] (a1) at (15,2){\phantom{u}};
    \node[circle] (a2) at (15,15){\phantom{u}};
    \node[circle] (a3) at (25,2){\phantom{u}};
    \node[circle] (a4) at (25,15){\phantom{u}};
    \node[circle] (b1) at (19.5,12){};
    \node[circle] (b2) at (20.5,12){};
    \draw[-latex] (a1)--(v1);
    \draw[-latex] (a2) --(b1);
    \draw[-latex] (v2) --(a3);
    \draw[-latex] (b2) --(a4);
    \draw[-latex] (v1) to (u);
    \draw[-latex] (u) to (v2);
    \draw[-latex] (v1) to (v2);
  \end{tikzpicture}
  \caption*{ }
\end{minipage}
    \caption{$S$ is the strong component in $D$ containing the $2$-cycle $uv$}\label{2linkred}
\end{figure}

By the description of the strong components above, $S$ contains a 2-cycle $uvu$ such that $v$ has in- and out-degree one in $S$. Let $D^*$ be obtained from $D$ by deleting $v$, adding  two new vertices $v_1,v_2$, adding an arc $wv_1$   for each arc $wv$ with $w\neq u$ in $D$, adding and arc $v_2h$ for each arc $vh$ with $h\neq u$ in $D$
 and finally adding the arcs $v_1u,uv_2,v_1v_2$, see Figure \ref{2linkred}. It is easy to check that $D^*$ has circumference at most 2 and that $D^*$ has the desired paths if and only if $D$ does: the vertex $v$ can be part of at most one path in a solution in $D$ and at most one path in a solution for $D^*$ can intersect the set $\{v_1,v_2\}$.

Now it is clear that by repeatedly processing  one leaf vertex at a time from a non-trivial strong component, as long as one exists, 
we obtain an equivalent acyclic instance after at most $|V(D)|-2k$ steps. 
\end{proof}

As mentioned in the introduction, the weak $k$-linkage problem is the arc-disjoint version of the $k$-linkage problem, where we ask for arc-disjoint rather than vertex-disjoint paths $P_1,\ldots{},P_k$ such that $P_i$ is an $(s_i,t_i)$-path for $i\in [k]$. Note that now the digraphs may have parallel arcs and a vertex from 
$\{s_i,t_i\}$ may be a vertex of one or more paths $P_j$ with $j\neq i$. Let us denote by $\mu_D(u,v)$ the number of arcs from $u$ to $v$ in the digraph $D$.

For acyclic digraphs the weak $k$-linkage problem is polynomial for every fixed $k$ \cite{fortuneTCS10}.
For digraphs for which all closed trails have length at most $p$ (e.g. when all strong components have size at most $p$) we can obtain a polynomial algorithm.

\begin{theorem}
\label{boundedwalk}
For every natural number $p$ the weak $k$-linkage problem is polynomial for digraphs containing no closed trail of length more than $p$. 
\end{theorem}

\begin{proof}
Let $[D,s_1,\ldots{},s_k,t_1,\ldots{},t_k]$ be an instance of the weak $k$-linkage problem where $D$ has no closed trail longer than $p$. Form a new digraph $D'$ by adding $2k$ new vertices $s'_1,\ldots{},s'_k,t'_1,\ldots{},t'_k$ and the arcs $\{s'_is_i|i\in [k]\}\cup\{t_it'_i|i\in [k]\}$. Clearly $[D',s'_1,\ldots{},s'_k,t'_1,\ldots{},t'_k]$ is a yes-instance if and only if $[D,s_1,\ldots{},s_k,t_1,\ldots{},t_k]$ is a yes-instance. Now let $L(D')$ be the line digraph of $D'$, that is, $V(L(D'))=A(D')$ and $A(L(D'))=\{ab|a,b\in A(D')\mbox{ and the head of $a$ coincides with the tail of $b$}\}$. It is easy to see that  $[D',s'_1,\ldots{},s'_k,t'_1,\ldots{},t'_k]$ is a yes-instance for the weak $k$-linkage problem if and only if $[L(D'),s'_1s_1,\ldots{},s'_ks_k,t_1t'_1,\ldots{},t_kt'_k]$ is a yes-instance for the $k$-linkage problem. Since $D$ and hence also $D'$ has no closed trail of length more than $p$, the circumference of $L(D')$ is bounded by $p$ and the claim now follows from Theorem \ref{klinkboundedcirc}.
\end{proof}

Note that we cannot apply the reduction above if we replace  the assumption of bounded maximum length of a closed  trail by bounded circumference since the circumference of $L(D')$ may be arbitrarily large compared to that of $D'$ (on the other hand, if $D$ is acyclic, then so is $L(D')$ and hence (as is well known) the weak $k$-linkage problem for acyclic digraphs reduces to the $k$-linkage problem in the same class in linear time).

\begin{problem}
What is the complexity of the weak $k$-linkage problem for digraphs with bounded circumference? 
\end{problem}

When the circumference is 2 we can give a polynomial algorithm.

\begin{theorem}
\label{weaklcirc2}
The weak $k$-linkage problem is polynomially solvable for every fixed $k$ in digraphs of circumference 2.
\end{theorem}

\begin{proof}
Like for the vertex disjoint case we will reduce an instance $[D,s_1,\ldots{},s_k,t_1,\ldots{},t_k]$ to an equivalent instance 
$[D',s'_1,\ldots{},s'_k,t'_1,\ldots{},t'_k]$ where $D'$ is an acyclic digraph. As above we show how to reduce the number of cycles in $D$ successively while preserving an equivalent instance until we have an equivalent acyclic instance and then we can apply the polynomial algorithm from \cite{fortuneTCS10} for acyclic digraphs.

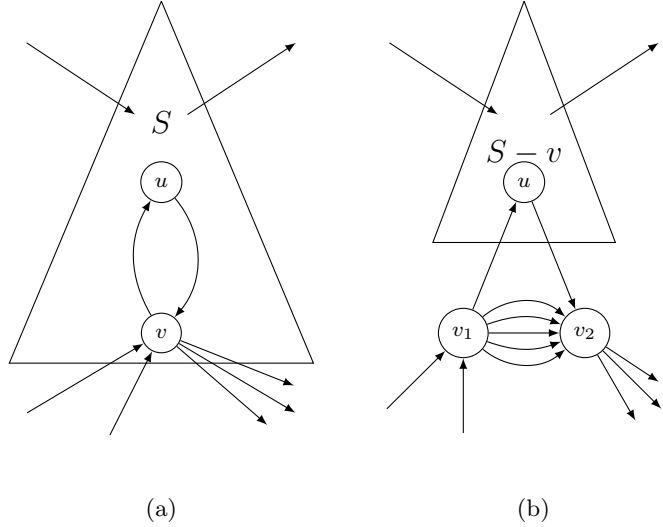
\begin{figure}[h]
\centering
  \begin{subfigure}[b]{0.3\textwidth}
   \centering
  \begin{tikzpicture}[scale=0.4]
  \tikzstyle{every node}=[font=\footnotesize]
  \node[circle, draw] (v) at (5,5){$v$};
  \draw (5,16) -- (0,4) --(10,4)--(5,16);
  \node[font=\large] at (5,12){$S$};
  \node[circle,draw] (u) at (5,10){$u$};
  \node[circle] (a1) at (0,2){\phantom{u}};
  \node[circle] (a11) at (3,1){\phantom{u}};
  \node[circle] (a111) at (1.5,1.5){\phantom{u}};
  \node[circle] (a2) at (0,15){\phantom{u}};
  \node[circle] (a3) at (10,2){\phantom{u}};
  \node[circle] (a33) at (10,3){\phantom{u}};
  \node[circle] (a333) at (8,1){\phantom{u}};
  \node[circle] (a3333) at (9,1.5){\phantom{u}};
  \node[circle] (a4) at (10,15){\phantom{u}};
  \node[circle] (b1) at (4.5,12){};
  \node[circle] (b2) at (5.5,12){};
  \draw[-latex] (a1)--(v);
  \draw[-latex] (a11)--(v);
  \draw[-latex] (a2) --(b1);
  \draw[-latex] (v) --(a3);
  \draw[-latex] (v) --(a33);
  \draw[-latex] (v) --(a3333);
  \draw[-latex] (b2) --(a4);
  \draw[-latex] (v) to [in=240, out=120] (u);
  \draw[-latex] (u) to [in=50, out=310] (v);
  \end{tikzpicture}
  \caption{}
  \end{subfigure}
  \begin{subfigure}[b]{0.3\textwidth}
\begin{tikzpicture}[scale=0.4]
  \tikzstyle{every node}=[font=\footnotesize]
  \node[circle, draw] (v1) at (3,5){$v_1$};
  \node[circle, draw] (v2) at (7,5){$v_2$};
  \draw (5,16)--(2,8)--(8,8)--(5,16);
  \node[font=\large] at (5,11) {$S-v$};
  \node[circle,draw] (u) at (5,10){$u$};
  \node[circle] (a1) at (0,2){\phantom{u}};
  \node[circle] (a11) at (3,1){\phantom{u}};
  \node[circle] (a111) at (1.5,1.5){\phantom{u}};
  \node[circle] (a2) at (0,15){\phantom{u}};
  \node[circle] (a3) at (10,2){\phantom{u}};
  \node[circle] (a33) at (10,3){\phantom{u}};
  \node[circle] (a333) at (8,1){\phantom{u}};
  \node[circle] (a3333) at (9,1.5){\phantom{u}};
  \node[circle] (a4) at (10,15){\phantom{u}};
  \node[circle] (b1) at (4.5,12){};
  \node[circle] (b2) at (5.5,12){};
  \draw[-latex] (a1)--(v1);
  \draw[-latex] (a11)--(v1);
  \draw[-latex] (a2) --(b1);
  \draw[-latex] (v2) --(a3);
  \draw[-latex] (v2) --(a33);
  \draw[-latex] (v2) --(a3333);
  \draw[-latex] (b2) --(a4);
  \draw[-latex] (v1) to (u);
  \draw[-latex] (u) to (v2);
  \draw[-latex] (v1) to (v2);
  \draw[-latex] (v1) to [in=160, out=20](v2);
  \draw[-latex] (v1) to [in=200, out=340] (v2);   
  \draw[-latex] (v1) to [in=140, out=40](v2);
  \draw[-latex] (v1) to [in=220, out=320] (v2);   
\end{tikzpicture}
\caption{}
  \end{subfigure}
\caption{The operation on the graph $D$, to split of a leaf vertex of a non-trivial strong component.}\label{weak2linkfig}
\end{figure}

Again we consider a non-trivial strong component $S$ and fix a 2-cycle $uvu$ such that $u$ is the only neighbour of $v$ in $S$ (there may be several arcs in both directions between $u$ and $v$). Let $d=d^-_{V-V(S)}(v)+d^+_{V-V(S)}(v)$, that is, the total number of arcs with one end in $v$ and the other in $V-V(S)$. 
Let $D^*$ be obtained from $D$ by replacing $v$ by two new vertices $v_1,v_2$,
 adding  $d$ arcs from $v_1$ to $v_2$ and replacing  the arcs incident with $v$ in $D$ by new arcs incident to $v_1,v_2$ as follows: for every vertex $w\not\in\{u,v\}$ add $\mu_D(w,v)$ arcs from $w$ to $v_1$ and $\mu_D(v,w)$ arcs from $v_2$ to $w$. Finally add $\mu_D(v,u)$ arcs from $v_1$ to $u$ and $\mu_D(u,v)$ arcs from $u$ to $v_2$, see Figure \ref{weak2linkfig}

It is easy to see that no new cycle is created when going from $D$ to $D^*$ and the 2-cycle $uvu$ disappears. It is also easy to verify that $[D^*,s'_1,\ldots{},s'_k,t'_1,\ldots{},t'_k]$ is a yes-instance if and only if $[D,s_1,\ldots{},s_k,t_1,\ldots{},t_k]$, where we let $s'_i=s_i$ ($t'_i=t_i$) 
if $s_i\neq v$ ($t_i\neq v$) and $s'_i=v_1$ if $s_i=v$ ($t'_i=v_2$ if $t_i=v$) : the $d$ arcs between $v_1$ and $v_2$ allow us to reroute any set of paths containing $v$ in $D$ (including paths starting or ending here) as paths in $D^*$ and conversely any solution in $D^*$ can be converted back to a solution in $D$. In particular, if a path $P'_i$ in $D^*$ contains the subpath $v_1uv_2$ we replace this part by the vertex $v$.

Thus after at most $|V|-1$ repetitions of the operation above we have converted
the original instance into an equivalent acyclic instance.
\end{proof}

\section{Concluding remarks}

Contrary to the case of  undirected graphs where a large class of NP-complete problems become polynomially solvable for graph of bounded tree-width, bounded DAG-width often does not lead to polynomial algorithms. For a discussion on this see \cite{ganianLNCS5917}.
The Hamiltonian cycle problem is polynomially solvable on digraphs of bounded directed tree-width \cite{johnsonJCT82} and hence, by Theorem \ref{ThmDirTree} also on digraphs of bounded DAG-width.

The minimum strong spanning subdigraph (MSSS) problem is the problem of deciding, for a given strong digraph $D=(V,A)$ a spanning strong subdigraph $D'=(V,A')$ with the minimum possible number of arcs. This problem is clearly NP-hard as for general digraphs as it contains the hamiltonian cycle problem as a special case. 
Khuller et al \cite{khullerDAM69} showed that the MSSS problem is NP-hard already for digraphs of circumference at most 5. Hence, by Theorem \ref{main}, we get.\\

\begin{theorem}
The MSSS problem is NP-hard for digraphs of DAG-width at most 5.
\end{theorem}

The following Erd\H{o}s-P\'osa type result was proved recently by Havet and Maia. 
\begin{theorem}\cite{havetsub}
\label{EPdigraphs}
There exists a function $f(k)$ such that every digraph $D$ either has $k$ disjoint cycles of length at least 3 or a set $X$ of size at most $f(k)$ such that $D-X$ has circumference at most 2.
\end{theorem}

From this result and Theorem \ref{main}  we obtain the following by letting $t(k)=f(k)+2$. 
\begin{corollary}
\label{noklongCD}
There exists a function
 $t(k)$ such that digraphs without 
$k$ disjoint cycles of length at least 3 have DAG-width at most $t(k)$.\qed
\end{corollary}

From Corollary \ref{noklongCD} and Theorem \ref{CorBoundDag} we obtain.

\begin{corollary}
The $k$-linkage problem is polynomially solvable for all fixed $k$ in the class of digraphs with at most $l$ disjoint cycles of length at least 3.
\end{corollary}

A similar result holds for the weak $k$-linkage problem but here more work is required.
 
\begin{theorem}
For every choice of natural numbers $k,l$ the weak $k$-linkage problem is solvable in polynomial time in the class of digraphs with no set of $l$ disjoint cycles all of length at least 3. 
\end{theorem}

\begin{proof} (sketch) 
Let $[D,s_1,\ldots{},s_k,t_1,\ldots{},t_k]$ be an instance of the weak $k$-linkage problem where $D$ is a digraph with no set of $l$ disjoint cycles of length at least 3 (it may contain arbitrarily many disjoint 2-cycles).
By Theorem \ref{EPdigraphs} there exists a set $X\subseteq V(D)$ of size at most $f(l)$ such that $D'=D-X$ has circumference 2 and we can find such a set in polynomial time (this follows from the proof in \cite{havetsub} but we could also just try all subsets of size at most $f(l)$).

Now suppose $[D,s_1,\ldots{},s_k,t_1,\ldots{},t_k]$ is a yes-instance and that
$P_1,\ldots{},P_k$ is a set of arc-disjoint paths forming a solution. Each $P_i$
can visit $X$ between $0$ and $|X|$ times, the first case corresponding to the path either staying completely inside $X$ or avoiding it altogether. We will not cover all the possibilities, but the idea should be clear from the description below. Suppose below that every $P_i$ starts and ends in $V-X$ (that is, $\{s_1,\ldots{},s_k,t_1,\ldots{},t_k\}\cap X=\emptyset$) and visits $X$ some number $r_i$ with  $0\leq r_i\leq |X|$ times (the case when some pairs $s_i,t_i$ intersect $X$ is easily adapted from the case below). If $r_i=0$ $P_i$ is just an $(s_i,t_i)$-path in $D'$ so suppose $r_i\geq 1$. Then $P_i=P_{i,1}a_{i,1}Q_{i,1}a'_{i,1}P_{i,2}a_{i,2}\ldots Q_{i,r}a'_{i,r_i}P_{i,r_i}$ where $a_{i,q}=u_{i,q}\tilde{u}_{i,q}$, $a'_{i,q}=\tilde{v}_{i,q}v_{i,q}$ are arcs respectively from $V-X$ to $X$ and from $X$ to $V-X$, $P_{i,1}$ is an $(s_i,u_{i,1})$-path, $P_{i,r_i}$ is a $(v_{i,r_i},t_i)$-path, $P_{i,j}$, $j\in [r_i]$,
is a $(v_{i,j-1}u_{i,j})$ path all in $D'$ and each $Q_{i,j}$, $j\in [r_i]$, is a $(\tilde{u}_{i,j},\tilde{v}_{i,j})$-path in the subdigraph $\induce{D}{X}$ induced by $X$. 
So  $P_i$ corresponds to $r_i+1$ arc-disjoint paths in $D'$ and $r_i$ arc-disjoint paths in $\induce{D}{X}$. Similarly for the remaining paths $P_j$. Thus the solution $P_1,\ldots{},P_k$ in $D$ corresponds to a solution to a weak $r$-linkage problem in $D'$, where $r$ represents the total number of subpaths of the $P_i$'s that lie inside $D'$ and a solution to a weak linkage problem in $X$ (the paths of the form $Q_{i,j}$). Clearly the converse also holds, if we have such a collection of paths in $D'$ and $\induce{D}{X}$ and the appropriate set of distinct arcs going to and from $X$, then $[D,s_1,\ldots{},s_k,t_1,\ldots{},t_k]$ is a yes-instance. Since  $\induce{D}{X}$ has bounded size, we can check any combination of paths here in constant time. Also $r$ cannot be larger than $k|X|+k$ (since no $P_i$ visits $X$ more than $|X|$ times) and hence the corresponding weak $r$-linkage problem for $D'$ can be solved in polynomial time via the algorithm from Theorem \ref{weaklcirc2}.

Hence by considering all possible choices of $r_i$ arc pairs $a_{i,j},a'_{i,j}$ for $i\in [k]$, $j\in [r_i]$ (all distinct) and solving each of the corresponding weak linkage problems we obtain a polynomial algorithm for the weak $k$-linkage problem (there are at most $|X|^{2k}$ different choices for the $r_i$'s, each involving at most $2k|X|$ arcs between $X$ and $V-X$ and hence it suffices to solve a polynomial number of weak linkage problems in $D'$).
\end{proof}

\bibliography{refs}
\section{Appendix}


To illustrate how to obtain a DAG-decomposition from the strategy for the cops used in the proof of Theorem \ref{main} we will now give an example to show how 
the DAG decomposition can be obtained by using this strategy. In figure \ref{exampleGraphD} the digraph $D$ on 32 vertices is shown. $D$ has circumference $4$ and consist of six strong components. The vertex sets of the six  strong components (listed according to an acyclic ordering) are
$S_1=\{1,2\}$, $S_2=\{27,28\}$, $S_3=\{5,6\}$, $S_4=\{3,4\}$, $S_5=\{7,8,9,10,11,12,13,14,15,16,17,18\}$ and $S_6=\{19,20,21,22,23,24,25,26,29,30,31,32\}$, and the strong component digraph $SC(D)$ (obtained by contracting each strong component into a vertex) is

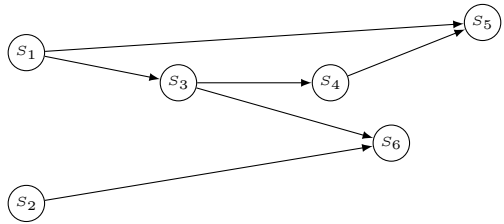
\begin{figure}[h]
\centering
   \begin{tikzpicture}[scale=0.4]
    \tikzstyle{every node}=[font=\tiny,circle, draw, inner sep=1.5pt]
    \node (v1) at  (0,5){$S_1$};
    \node (v2) at  (0,0){$S_2$};
    \node (v3) at  (5,4){$S_3$};
    \node (v4) at  (10,4){$S_4$};
    \node (v5) at  (15,6){$S_5$};
    \node (v6) at  (12,2){$S_6$};
    \draw[->, style=-latex] (v1) to (v3);
    \draw[->, style=-latex] (v1) to (v5);
    \draw[->, style=-latex] (v2) to (v6);
    \draw[->, style=-latex] (v3) to (v4);
    \draw[->, style=-latex] (v3) to (v6);
    \draw[->, style=-latex] (v4) to (v5);
   \end{tikzpicture}
\end{figure}

Now for each strong component we use the strategy described in the proof of Theorem \ref{main}. Notice that $S_1$, $S_2$, $S_3$ and $S_4$ all consist of $2$ vertices and hence here, in turn, we just place a cop on each vertex and after this we have chased the robber into either $S_5$ or $S_6$. Now in $\induce{D}{S_5}$ a maximal  cycle is formed by the vertices $\{15,16,17,18\}$ and we place a cop on each of these vertices. Now either the robber is in $\{7,8,11,12\}$ or in $\{9,10,13,14\}$. In both cases we need only keep the cop on vertex $15$ to keep the robber in the same strong component. Say the robber is in $\{7,8,11,12\}$ then in the next step we place a cop on $12$, and note that 
the maximal  cycle containing $\{12,15\}$ is the 2-cycle formed by these vertices. Now the cop in vertex $15$ can be lifted and placing a cop on $11$ and finding a maximal cycle containing $11,12$ gives the cycle $\{7,8,11,12\}$ witch means that we have caught the robber and are done. The argument is symmetric for $\{9,10,13,14\}$. The DAG-decomposition of this strong component is seen in the subtree with root $\{15,16,17,18\}$ in the total DAG-decomposition in figure \ref{exampleDagOfD}. For the DAG-decomposition of $S_6$ we start by placing four cops on  the maximal cycle $C$ formed by $\{19,20,23,24\}$. Placing the cops on $C$ forces the robber to move to one of the 
two strong components $\{29,30\}$ and $\{21,22,25,26,31,32\}$ of $D-V(C)$. For the first of these the cops in vertex $\{19,20,23\}$ are lifted and one is placed on $29$ and in next step the cop on $24$ is lifted and placed on $30$ and we are done. With the second strong component it is instead the cops on vertex $19,23,24$ that are lifted and one is placed on  $21$. Then the cop on $20$ is lifted and cops are placed on $22,25,26$. The last two steps are equivalent to what we did after covering $\{19,20,23,24\}$.	

Now that we have found these DAG decompositions of each of the strong components, we can combine the DAG decomposition and $SC(D)$ to obtain the DAG decomposition for the hole graph. This is seen in figure \ref{exampleDagOfD}

\begin{figure}[h]
\centering
\begin{subfigure}[b]{0.2\textwidth}
\centering
  \begin{tikzpicture}[scale=0.2]
  \tikzstyle{every node}=[font=\tiny,circle, draw, fill,inner sep=1.5pt]
  \node[label=above left:{$1$}] (v1) at  (5,36){};
  \node[label=above right:{$2$}] (v2) at  (9,36){};
  \node[label=above left:{$3$}] (v3) at  (16,36){};
  \node[label=above right:{$4$}] (v4) at  (20,36){};
  \node[label=left:{$5$}] (v5) at  (5,32){};
  \node[label=right:{$6$}] (v6) at  (20,32){};
  \node[label=above left:{$7$}] (v7) at  (8,29){};
  \node[label=above right:{$8$}] (v8) at  (10,29){};
  \node[label=below left:{$11$}] (v11) at  (8,27){};
  \node[label=below:{$12$}] (v12) at  (10,27){};
  \node[label=above left:{$9$}] (v9) at  (15,29){};
  \node[label=above right:{$10$}] (v10) at  (17,29){};
  \node[label=below:{$13$}] (v13) at  (15,27){};
  \node[label=below right:{$14$}] (v14) at  (17,27){};
  \node[label=above:{$15$}] (v15) at  (12.5,25){};
  \node[label=right:{$17$}] (v17) at  (14,23.5){};
  \node[label=left:{$16$}] (v16) at  (11,23.5){};
  \node[label=below:{$18$}] (v18) at  (12.5,22){};
  \node[label=above left:{$19$}] (v19) at  (0,20){};
  \node[label=above left:{$20$}] (v20) at  (5,20){};
  \node[label=below left:{$23$}] (v23) at  (0,17){};
  \node[label=below left:{$24$}] (v24) at  (5,17){};
  \node[label=above right:{$21$}] (v21) at  (20,20){};
  \node[label=above right:{$22$}] (v22) at  (25,20){};
  \node[label=below right:{$25$}] (v25) at  (20,17){};
  \node[label=below right:{$26$}] (v26) at  (25,17){};
  \node[label=above:{$27$}] (v27) at  (10,5){};
  \node[label=above:{$28$}] (v28) at  (15,5){};
  \node[label=below:{$29$}] (v29) at  (5,0){};
  \node[label=below:{$30$}] (v30) at  (10,0){};
  \node[label=below:{$31$}] (v31) at  (15,0){};
  \node[label=below:{$32$}] (v32) at  (20,0){};
  \draw[] (v1) to (v2);
  \draw[] (v3) to (v4);
  \draw[->, style=-latex] (v1) to (v5);
  \draw[->, style=-latex] (v2) to (v8);
  \draw[->, style=-latex] (v3) to (v9);
  \draw[->, style=-latex] (v6) to (v4);
  \draw[] (v5) to (v6);
  \draw[] (v7) to (v8);
  \draw[] (v7) to (v11);
  \draw[] (v8) to (v12);
  \draw[] (v11) to (v12);
  \draw[] (v9) to (v10);
  \draw[] (v9) to (v13);
  \draw[] (v10) to (v14);
  \draw[] (v13) to (v14);
  \draw[] (v12) to (v15);
  \draw[] (v13) to (v15);
  \draw[] (v15) to (v16);
  \draw[] (v15) to (v17);
  \draw[] (v15) to (v18);
  \draw[] (v16) to (v17);
  \draw[] (v16) to (v18);
  \draw[] (v17) to (v18);
  \draw[->, style=-latex] (v5) to (v20);
  \draw[<-, style=-latex] (v6) to (v21);
  \draw[] (v19) to (v20);
  \draw[] (v19) to (v23);
  \draw[] (v20) to (v24);
  \draw[] (v23) to (v24);
  \draw[] (v21) to (v22);
  \draw[] (v21) to (v25);
  \draw[] (v25) to (v26);
  \draw[] (v22) to (v26);
  \draw[] (v20) to (v21);
  \draw[] (v24) to (v29);
  \draw[] (v25) to (v32);
  \draw[] (v29) to (v30);
  \draw[->, style=-latex] (v27) to (v30);
  \draw[] (v27) to (v28);
  \draw[->, style=-latex] (v28) to (v31);
  \draw[] (v31) to (v32);
  \end{tikzpicture}
\caption{A digraph $D$. Undirected edges correspond to 2-cycles.}\label{exampleGraphD}
\end{subfigure} \qquad \qquad \qquad
\begin{subfigure}[b]{0.6\textwidth}
\centering
  \begin{tikzpicture}[scale=0.3]
  \tiny
  \tikzstyle{every node}=
  [
  circle, 
  draw, 
  text width=2.5em,
  text centered,
  minimum height=3em,
  anchor=mid, 
  ]
  \node (root) at  (20,42){$\emptyset$};
  \node (1) at  (15,40){27 28};
  \node (2) at  (25,40){1 2};
  \node (3) at  (20,38){5 6};
  \node (4) at  (20,32){3 4};
\node (5) at  (25,30){15 16 17 18};
  \node (6) at  (22,25){12 15};
  \node (7) at  (28,25){13 15};
  \node (8) at  (22,20){7 8 11 12};
  \node (9) at  (28,20){9 10 13 14};
  \node (10) at (13,32){19 20 23 24};
  \node (11) at  (10,27){24 29};
  \node (12) at  (10,22){29 30};
  \node (13) at  (16,27){20 21};
  \node (14) at  (16,23){21 22 25 26 };
  \node (15) at  (16,19){25 32};
  \node (16) at  (16,15){31 32};


  \draw[->, style=-latex] (root) to (1);
  \draw[->, style=-latex] (root) to (2);
  \draw[->, style=-latex] (2) to (3);
  \draw[->, style=-latex] (3) to (4);
  \draw[->, style=-latex] (4) to (5); 
  \draw[->, style=-latex] (2) to (5);
  \draw[->, style=-latex] (5) to (6);
  \draw[->, style=-latex] (5) to (7);
  \draw[->, style=-latex] (6) to (8);
  \draw[->, style=-latex] (7) to (9);

  \draw[->, style=-latex] (1) to (10);
  \draw[->, style=-latex] (10) to (11);
  \draw[->, style=-latex] (10) to (13);
  \draw[->, style=-latex] (11) to (12);
  \draw[->, style=-latex] (13) to (14);
  \draw[->, style=-latex] (14) to (15);
  \draw[->, style=-latex] (15) to (16);

  \draw[->, style=-latex] (3) to (10);
  \end{tikzpicture}        
  \caption{A DAG decomposition of $D$ obtained via the cops and robber game.}\label{exampleDagOfD}
\end{subfigure}
\caption{}
\end{figure}

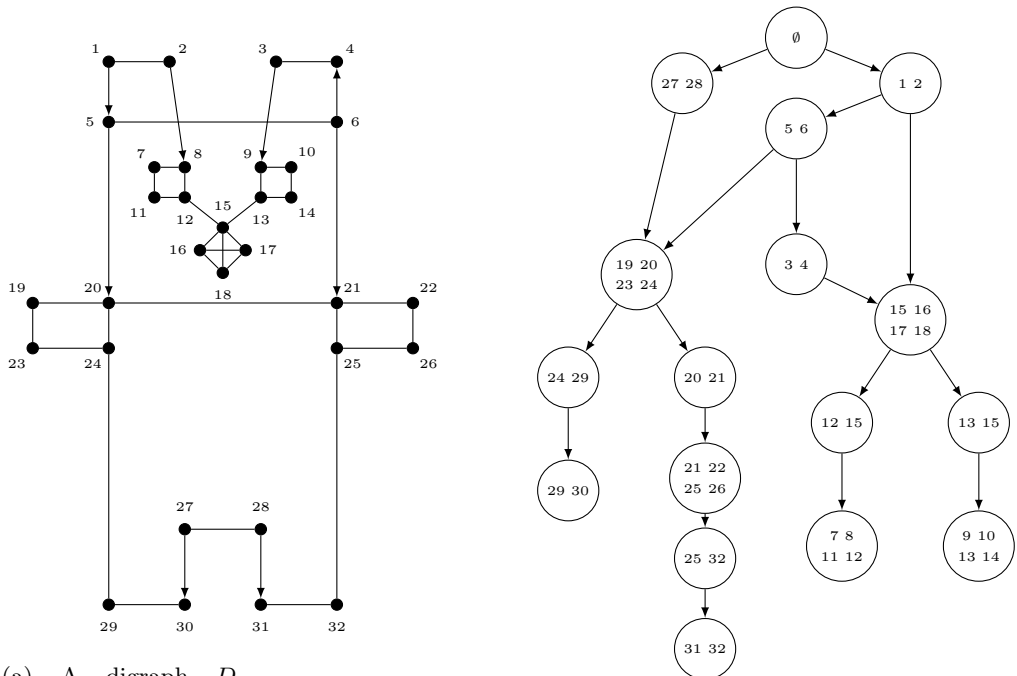
\end{document}